\documentclass[lettersize,journal]{IEEEtran}
\usepackage{amsmath,amsfonts}
\usepackage{algorithmic}
\usepackage{algorithm}
\usepackage{array}
\usepackage[caption=false,font=normalsize,labelfont=sf,textfont=sf]{subfig}
\usepackage{textcomp}
\usepackage{stfloats}
\usepackage{url}
\usepackage{longtable}
\usepackage{verbatim}
\usepackage{graphicx}
\usepackage{cite}
\usepackage{bm}       
\usepackage{amsthm}
\usepackage{color}
\newtheorem{thm}{Theorem}
\newtheorem{lemma}{Lemma}
\theoremstyle{definition}
\newtheorem{defn}{Definition}

\begin{document}


\title{A symmetric LWE-based Multi-Recipient Cryptosystem}

\author{Saikat Gope, Srinivasan Krishnaswamy, Chayan Bhawal
}



\maketitle

\begin{abstract}
This article describes a post-quantum multi-recipient symmetric cryptosystem whose security is based on the hardness of the LWE problem. In this scheme a single sender encrypts multiple messages for multiple recipients generating a single ciphertext which is broadcast to the recipients. Each recipient decrypts the ciphertext with her secret key to recover the message intended for her. In this process, the recipient cannot efficiently extract any information about the other messages. This scheme is intended for messages like images and sound that can tolerate a small amount of noise. This article introduces the scheme and establishes its security based on the LWE problem. Further, an example is given to demonstrate the application of this scheme for encrypting multiple images.
     
\end{abstract}


\section{Introduction}

\IEEEPARstart{A} Multi-Recipient Encryption Scheme (MRES) simultaneously encrypts multiple messages for multiple receivers. 
The idea of MRES was introduced by Bellare et al. \cite{bellare2000}. This scheme was a public key encryption scheme wherein a set of messages intended for different users are simultaneously encrypted using their respective public keys. The encryption algorithm generates a set of ciphertexts, one for each recipient. Further, Kurosawa proposed an MRES design having a shortened ciphertext \cite{kurosawa2002multi}. The ciphertext size is almost half of the $n$ recipient scheme and showed its security to be almost the same as that of a single-recipient scheme. MRES using randomness re-use was proposed in \cite{bellare2002randomness-reuse} for reducing transmission load and computational cost.  Other notable contributions in this area include \cite{bellare2003multi,bellare2007multi,bellare2007multirecipient}. \par
With the advancement of quantum computation and its applications \cite{shor1999polynomial,long2001grover}, problems such as the discrete logarithm problem or factorization of primes can be efficiently solved. Hence, encryption schemes whose security is based on the hardness of these problems become vulnerable to quantum attacks. Lattice-based problems serve as promising alternatives to develop schemes that are resistant to quantum attacks. Some hard problems related to lattices are the $\text{GapSVP}_{\gamma}$ problem, $\text{GapCVP}_{\gamma}$ problem and the SIVP problem \cite{Peikart2022_approx_SVP,micciancio2001_Gap_SVP_CVP,peikert2009_SVP_GapSVP,SIVP_SVP_CVP}. The Learning With Error (LWE) problem \cite{regev2009lattices,regev2010learning} involves solving a set of linear equations over a large finite field in the presence of noise.  The hardness of the LWE problem reduces to that of the $\text{GapSVP}_{\gamma}$ problem \cite{Peikart2022_approx_SVP,micciancio2001_Gap_SVP_CVP,peikert2009_SVP_GapSVP}. This problem has therefore led to the development of numerous post-quantum cryptosystems.

\par
In this paper, we implement a symmetric Multi-Key Multi-Recipient (MKMR) encryption scheme based on the hardness of the LWE problem. In this scheme the sender simultaneously encrypts a set of messages for a set of receivers. each receiver has a secret key shared with the sender. 
The sender generates {\bf a single ciphertext}, which each user decrypts using their respective secret key to recover their intended message.
We prove that the resulting tuple of ciphertexts is indistinguishable from a random block of data sampled from a uniform distribution of the appropriate size. Further, we demonstrate an application of this scheme for multiple image encryption. \par
This paper is organized into 4 sections. Section 2 introduces the preliminaries that are needed to understand the rest of the paper. In section 3, we explain the LWE-based multi-recipient encryption scheme along with an example. The conclusion is presented in section 4.

\subsection{Motivation and Contribution}

Most multi-recipient encryption schemes available in literature are public key schemes \cite{bellare2000,kurosawa2002multi,bellare2002randomness-reuse,bellare2003multi,bellare2007multi,bellare2007multirecipient}. Further, the security of these schemes are based on the hardness of problems like the discrete logarithm problem. These problems are efficiently solvable using quantum algorithms. Hence the available multi-recipient schemes are vulnerable in post-quantum scenario. 
This paper proposes a post-quantum symmetric multi-recipient encryption scheme whose security is based on the hardness of the LWE problem. A unique feature of this scheme is that it generates a single ciphertext for all messages. Each recipient uses their respective key to decrypt the ciphertext and recover the message intended for them. Sending the same ciphertext to multiple recipients introduces redundancy; if a recipient loses the ciphertext, they can recover it from another recipient. The proposed scheme is computationally light, with the encryption process consisting only of repeated matrix vector multiplications.



\section{Preliminaries}
\subsection{Notation Table}
Table 1 shows the set of notations that are used in the paper.
\begin{table}[!t]
\caption{Table of Symbolic Notations\label{tab:notations_tab}}
\centering
\begin{tabular}{|p{1.5cm}|p{6cm}|}
\hline
\textbf{Symbols Used} & \textbf{Interpretations}\\
\hline

$\mathbb{F}_q$ & Finite field $\mathbb{F}_q$ with cardinality $q$\\ 

\hline
$\mathbb{F}_q^n$  & n-dimensional vector space over finite field $\mathbb{F}_q$ with cardinality $q$\\

\hline
$\left \| v \right \|_{p}$ & $l_p$ norm of n-dimensional vector $v$ over field $\mathbb{F}$ with $p \geq 1$\\

\hline
$\left \| v \right \|$ & $2$ norm of n-dimensional vector $v$ over field $\mathbb{F}$ \\

\hline
$\left \langle s,v \right \rangle$ & Inner product of vector $s$ and vector $v$ where $s,v \in \mathbb{F}_q^n$ \\








\hline
$\left \lceil v \right \rfloor$ & Round elements of array $v$ to its nearest integer\\





\hline
$v$ $mod$ $q$ & The integer between $\lfloor \frac{-q}{2}\rfloor$ and $\lfloor \frac{q}{2}\rfloor$ which is equivalent to the integer $v$ modulo $q$ in $\mathbb{F}_q$\\

\hline
$\mathcal{M}_i$ & $i^{th}$ column of matrix $\mathcal{M}_{m \times l}$ where $i=1,2,\ldots,l$\\

\hline
$\chi^m$ & Set of $m$ tuple elements from distribution $\chi$ \\


\hline

\end{tabular}
\end{table}
We now formally define lattices and some of the hard problems related to lattices include the approximate shortest vector problem i.e. the $\text{GapCVP}_{\gamma}$ problem and Learning With Error (LWE) problem.

\subsection{Lattices}
\begin{defn}
Given $n$-linearly independent basis vectors $\mathcal{B} = \{ b_1, b_2, \ldots, b_n\} \subset \mathbb{R}^n$, the set of all integer linear combination of basis vectors $\mathcal{B}$ is defined as lattice $\mathcal{L} \subset \mathbb{R}^n $. Mathematically, it is represented as:
\begin{equation}
\label{lattice_L}
    \mathcal{L} := \mathcal{L}(\mathcal{B}) = \left \{\sum_{i=1}^{n}a_{i}b_{i} \mid a_{i} \in \mathbb{Z}, 1\leq i \leq n \right \}
\end{equation}
\end{defn}
\begin{defn}[\cite{peikert2009_SVP_GapSVP}]
\label{discrete_gaussian_definition}
Consider an $n$ dimensional lattice $\mathcal{L}$ with basis vectors $\mathcal{B}$ as in Definition \ref{lattice_L}, the discrete Gaussian probability distribution over lattice $\mathcal{L}$ with standard deviation $\sigma > 0$ is defined as
\begin{equation}
\label{discrete_gaussian}
    \mathcal{D}_{\mathcal{L},\sigma}(v) := \frac{\rho_\sigma(v)}{\rho_\sigma(\mathcal{L})}
\end{equation}
where $\rho_\sigma(v) := e^{\left(-\pi\left\| v\right\|^{2}/2\pi\sigma^2 \right)}$ for all $v \in \mathcal{L}$ and $\rho_\sigma(\mathcal{L}) := \sum_{y\in \mathcal{L}}e^{\left(-\pi\left\| y\right\|^{2}/2\pi\sigma^2 \right)}$.
\end{defn}
\begin{defn}
[\cite{Peikart2022_approx_SVP,peikert2009_SVP_GapSVP}]
Given $n$ linearly independent basis vectors $\mathcal{B} = \{ b_1,b_2, \ldots,b_n\} \subset \mathbb{R}^n$ over an $n$-dimensional lattice $\mathcal{L} = \mathcal{L}(\mathcal{B})$ , the minimum distance of the lattice $\mathcal{L}$ in $l_p$ norm with $p \geq 1$ is defined as,
    \begin{equation}
        \lambda^{(p)}(\mathcal{L}) := \displaystyle{\min_{y \in \mathcal{L}\backslash{\{0}\}} \left \| y \right \|_{p}}
    \end{equation}
    For a distance threshold $s>0$, the $\text{GapSVP}_{\gamma}$ problem with $\gamma(n) \geq 1$ refers to the problem of determining whether $\lambda^{(p)}(\mathcal{L})$ is a YES instance or NO instance. These instances are defined as follows:
    \begin{itemize}
        \item YES instance: $\lambda^{(p)}(\mathcal{L}) \leq s$
        \item NO instance: $\lambda^{(p)}(\mathcal{L}) > \gamma s$
    \end{itemize}
\end{defn}
\subsection{LWE Problem}
    The learning with errors problem with parameters $n,q,\chi$ namely $\mathcal{LWE}_{n,q,\chi}$, refers to solving a noisy set of linear equations over the finite field $\mathbb{F}_q$ where the noise is sampled from a distribution $\chi$ over $\mathbb{F}_q$. We now formally define two variants of $\mathcal{LWE}_{n,q,\chi}$ problem, the search $\mathcal{LWE}_{n,q,\chi}$ and the decision $\mathcal{LWE}_{n,q,\chi}$.
\subsubsection{Search $\mathcal{LWE}$}
\label{Search_lwe_prob_section}
\begin{defn}
Let, the secret vector $\bm{s}$ and vectors $a_1, a_2, \ldots, a_m$ be sampled from uniform distribution over $\mathbb{F}^n_{q}$. Given, the set of samples pairs $(a_1,b_1),(a_2,b_2), \ldots,(a_m,b_m)$, the search $\mathcal{LWE}_{n,q,\chi}$ refers to the problem of finding $s$ from the pair $(a_i,b_i)$ for $i = 1,2, \ldots,m$ such that $b_i \equiv \left \langle \bm{s},a_{i} \right \rangle + e_{i}\,\, ( \text{mod}\,\,q) $ where $e_{i}$ is sampled from discrete Gaussian distribution $\chi$ over $\mathbb{F}_{q}$. 
\end{defn}
\subsubsection{Decision $\mathcal{LWE}$}
\label{Decision_lwe_prob_section}
\begin{defn}
Let, the secret vector $\bm{s}$ and vectors $a_1, a_2, \ldots, a_m$ be sampled from uniform distribution over $\mathbb{F}^n_{q}$. Consider the set of vectors $(a_1,b_1),(a_2,b_2), \ldots,(a_m,b_m)$,  where $b_i \equiv \left \langle \bm{s},a_{i} \right \rangle + e_{i}\,\, ( \text{mod}\,\,q) $ and $e_{i}$ is sampled from discrete Gaussian distribution $\chi$ over $\mathbb{F}_{q}$. The decision $\mathcal{LWE}_{n,q,\chi}$ problem refers to the problem of distinguishing this $m$-tuple from a set of $m$ vectors that are randomly sampled from a uniform distribution over $\mathbb{F}_{q}^{n+1}$. 
\end{defn}
The decision $\mathcal{LWE}_{n,q,\chi}$ is as hard as the search $\mathcal{LWE}_{n,q,\chi}$ \cite{regev2009lattices}, and the search $\mathcal{LWE}_{n,q,\chi}$ problem is equivalent to $\text{GapSVP}_{\gamma}$ problem for large $q \geq 2^{n/2}$ \cite{peikert2009_SVP_GapSVP}. The distribution of the $m$-tuple $(a_1,b_1),(a_2,b_2), \ldots,(a_m,b_m)$ is called the $\mathcal{LWE}_{n,q,\chi}$ distribution and the assumption that the $\mathcal{LWE}_{n,q,\chi}$ is hard is referred to as the $\mathcal{LWE}_{n,q,\chi}$ assumption.\par
The LWE problem can be used to design a cryptosystem having multiple keys for multiple participants. It is explained in section \ref{LWE_multi_key_cryptosys}.

\section{LWE based Multi Recipient Cryptosystems
}
\label{LWE_multi_key_cryptosys}
In this section, we propose an LWE-based multi recipient encryption scheme. The proposed scheme considers a configuration with a single sender and multiple receivers. Each receiver has a secret key shared with the sender. Here, when the number of receivers is $m$, an $m$-tuple of message streams is encrypted using an $m$-tuple of secret keys (one corresponding to each receiver) to produce a single ciphertext. Each recipient can recover their intended message using their secret key.

 We start by defining a pseudorandom map, which is the building block for the multi-recipient scheme. 

\begin{defn}
\label{func_generated_map}
    A function $f_{s,\chi} : \mathbb{F}_q^m \rightarrow \mathbb{F}_q$ indexed by an arbitrary chosen vector $s \xleftarrow{\scriptscriptstyle{\$}} \mathbb{F}_q^m$ and a distribution $\chi$ is a function-generated map defined as
    \begin{equation}
        f_{s,\chi}\left(v\right) := \left(\left \langle s,v \right \rangle + e \right)\,\, ( \text{mod}\,\,q)
    \end{equation}
    where $v \xleftarrow{\scriptscriptstyle{\$}} \mathbb{F}_q^m$, $e \in \mathbb{F}_q$ is sampled from distribution $\chi$.
\end{defn}
\begin{defn}
\label{PRM_defn}
A map $f : S \times \mathbb{F}_q^m \rightarrow \mathbb{F}_q^m$ indexed by elements of set $S$ and distribution $\chi$ is said to be pseudorandom map (\textit{PRM}) if the following properties hold:


\begin{enumerate}
    \item The map $f$ is efficiently computable. 
    
    
    \item For an element $s \xleftarrow{\scriptscriptstyle{\$}} S$ the output of $f(s,\bullet)$ is computationally indistinguishable from a randomly sampled element from $\mathbb{F}_q^m$. In other words, the advantage of an adversary algorithm $\mathcal{A}$ with oracle access to $f$, in distinguishing the output of $f$ from a randomly sampled element of $\mathbb{F}_q^m$ is bounded by a negligible function $\epsilon$ of $m$.
    \begin{align}
        Adv(\mathcal{A}) &= \left|Pr\left[\mathcal{A}\left(f(s,.) \right)=1 \right] - Pr\left[\mathcal{A}\left(U(\mathbb{F}_q^m) \right)=1 \right] \right| \notag \\
        &\leq \epsilon(m)
    \end{align}
    
\end{enumerate}
\end{defn}
The only difference between a \textit{PRM} and a pseudorandom function (\textit{PRF}) is that, in a \textit{PRM} each input can potentially have multiple outputs.\par
{The following theorem demonstrates how a \textit{PRM} $f:S \times \mathbb{F}_q^m \rightarrow \mathbb{F}_q^m$ can be used recursively to construct a \textit{PRM} $\mathcal{F}:S \times \mathbb{F}_q^m \rightarrow \mathbb{F}_q^{m \times l}$. }
In the theorem below, we are introducing an arbitrary matrix $M$ which corresponds to the message in the encryption scheme described in Section \ref{sec_enc_scheme}.

\begin{thm}
\label{thm_PRM}
    Consider a PRM $f:S \times \mathbb{F}_q^m \rightarrow \mathbb{F}_q^m$. 
    Let $M = \left[ m_1,m_2,\ldots,m_l \right] \in \mathbb{F}_q^{m \times l}$, be an arbitrary matrix. 
    Define $\mathcal{F}_{M,S}: \mathbb{F}_q^m \rightarrow \mathbb{F}_q^{m \times l}$ as
    \begin{eqnarray}
    \label{eqn_PRM}
        \mathcal{F}_{M,S}(s,g_0) &=& (g_1,g_2,\ldots,g_l), \notag\\
        \textrm{ and } g_i &=& f(s,g_{i-1}) + m_i,
    \end{eqnarray}
    where  $g_0 \xleftarrow{\scriptscriptstyle{\$}} \mathbb{F}_q^m$ and $s$ is randomly sampled from uniform distributions over $S$, independent of each other and the columns of $M$. Then $\mathcal{F}_{M,S}$ is indistinguishable from a randomly sampled map from $\mathbb{F}_q^m \rightarrow \mathbb{F}_q^{m \times l}$.
\end{thm}

\begin{proof}
    Construct a series of hybrid experiments $H_0,H_1,\ldots,H_l$ where in the experiment $H_k$, the vectors $(g_0,g_1,\ldots,g_k) \in \mathbb{F}_q^{m \times (k+1)}$ are randomly sampled from uniform distribution over $\mathbb{F}_q^m$ and $g_i$ for $i = k+1,k+2,\ldots,l$ are calculated as follows:
    \begin{equation}
        g_i = f(s,g_{i-1}) + m_i
    \end{equation}
    The output of $H_k$ is $(g_1,\ldots,g_l) \in \mathbb{F}_q^{m \times l}$. {Observe that the output of the experiment $H_0$ corresponds to the output of map $\mathcal{F}_{M,S}$} and $H_l$ corresponds to the output of a map randomly sampled from a uniform distribution on the set of maps from $\mathbb{F}_q^m$ to $\mathbb{F}_q^{m \times l}$.\par
    Consider a binary algorithm $\mathcal{A}$ that accepts elements of $\mathbb{F}_q^{m \times l}$ as inputs and aims to distinguish between the experiments $H_0$ and $H_l$. Let 
    \begin{equation}
        \mathcal{P}_k = Pr\left[\mathcal{A}(H_k) = 1 \right]
    \end{equation}
    be the probability that $\mathcal{A}$ returns $1$ when it takes elements generated from the experiment $H_k$ as the input. If for any $k$, algorithm $\mathcal{A}$ can distinguish between samples from $H_{k-1}$ and $H_k$, then $\mathcal{A}$ can be used to construct an adversary $\mathcal{A}'$, with oracle access to $f(s,\bullet)$ that can distinguish outputs of $f(s,\bullet)$ from that of a truly random map. Let the input of $\mathcal{A}'$ be $(z,h)$  where $z$ is randomly sampled from $U(\mathbb{F}_q^m)$ and $h$ is either $f(s,z)$ or $f'(z)$ where $f'$ is randomly sampled from the set of maps from $\mathbb{F}_q^m$ to $\mathbb{F}_q^m$. Generate the vectors 
    \begin{align}
        g_k &= h + m_k \notag \\
        g_{k+1} &= f(s,g_k) + m_{k+1} \notag \\
        \vdots \notag \\
        g_l &= f(s,g_{l-1}) + m_l \notag \\
    \end{align}
    using oracle access to $f(s,\bullet)$. Further, randomly sample vectors $(g_1,\ldots,g_{k-2})$ from a uniform distribution on $\mathbb{F}_q^m$ and let $g_{k-1}=z$. Observe that if $h=f(s,z)$, then the distribution of the vectors $(g_1,\ldots,g_l)$ is identical to the distribution of the output of $H_{k-1}$. On the other hand, if $h=f'(z)$, then the distribution of this set of vectors is identical to that of the output of $H_{k}$. Thus, if $\mathcal{A}'$ can parse the vectors $(g_1,\ldots,g_l)$ to $\mathcal{A}$ which can distinguish between the output distributions of $H_k$ and $H_{k-1}$ with a significant advantage then, $\mathcal{A}'$ can distinguish between $f(s,\bullet)$ and a randomly sampled map. This contradicts the pseudorandomness of $f(s,\bullet)$. Hence, $\left|\mathcal{P}_{k} - \mathcal{P}_{k-1} \right| < \epsilon(m)$ where  $\epsilon$ is a function that is negligible in $m$.\par
    Now, the advantage of $\mathcal{A}$ in distinguishing between outputs of $H_0$ and $H_l$ is given by
    \begin{align}
        Adv(\mathcal{A}) &= \left|\mathcal{P}_{0} - \mathcal{P}_{l} \right| \notag \\
        &= \left|\mathcal{P}_{0} - \mathcal{P}_{1} + \mathcal{P}_{1} - \mathcal{P}_{2} + \ldots + \mathcal{P}_{l-1} - \mathcal{P}_{l} \right| \notag \\
        &\leq \left|\mathcal{P}_{0} - \mathcal{P}_{1} \right| + \left|\mathcal{P}_{1} - \mathcal{P}_{2} \right| + \ldots + \left|\mathcal{P}_{l-1} - \mathcal{P}_{l} \right| \notag \\
        &\leq l\epsilon(m)
    \end{align}
    Hence, $Adv(\mathcal{A})$ is negligible in $m$.
\end{proof}

The following subsection explains the construction of a pseudorandom map based on the hardness of the LWE problem.

\subsection{LWE-based Pseudorandom Map (LWE-PRM)}
In this section, we use the LWE problem for the construction of a \textit{PRM}. We refer to such a PRM as an {\em LWE-PRM}. Here, multiple secret vectors $s_k \in \mathbb{F}_q^m$ for $k=1,2,\ldots,m$ are sampled randomly and a vector $v \xleftarrow{\scriptscriptstyle{\$}}  \mathbb{F}_q^m$ randomly sampled such that the addition of noise $e \in \mathbb{F}_q$ sampled from distribution $\chi$ with the inner product of secret vector $s_k$ and the vector $v$ gives $m$ elements in $\mathbb{F}_q$.\par
The following theorem is a special case of Lemma 6.2 from Peikert et al. \cite{peikert2008lossy}, which is instrumental in the development of our results. This theorem proves the indistinguishability property of the \textit{LWE-PRM} from a randomly sampled map from $\mathbb{F}_q^m \rightarrow \mathbb{F}_q^m$. 
\begin{lemma}
\label{thm_LWE_PRM}
    Consider $S = \left[ s_1, s_2, \ldots, s_m \right]^T \in \mathbb{F}_q^{m \times m}$ with $s_k \xleftarrow{\scriptscriptstyle{\$}} \mathbb{F}_q^m$ for $k = 1,2, \ldots, m$ to be a random secret vector matrix and $v \xleftarrow{\scriptscriptstyle{\$}}  \mathbb{F}_q^m$ be a randomly chosen vector. Define the map $f_{\chi}\left(\bullet,v \right)$ such that 
    \begin{align}
    \label{f_S_prm}
        f_{\chi}\left(S,v \right) &:= (S \times v + E)\,\,\left(\text{mod}\,\,q\right) \\
        &=\begin{bmatrix}
        \left \langle s_1,v \right \rangle + e_1 \\
        \left \langle s_2,v \right \rangle + e_2 \\
        \left \langle s_3,v \right \rangle + e_3 \\
        \vdots  \\
        \left \langle s_m,v \right \rangle + e_m \\
        \end{bmatrix}\,\,\left(\text{mod}\,\,q\right) \notag
\end{align}
with error vector $E = \left[e_1,e_2,\ldots,e_m \right]^T \in \mathbb{F}_q^m$ and $e_k \in \mathbb{F}_q$ sampled from distribution $\chi$ for $k = 1,2,\ldots,m$. Then, under the $\mathcal{LWE}_{m,q,\chi}$ assumption$, f_{\chi}\left(\bullet,v \right)$ is a pseudorandom map.
\end{lemma}
\begin{proof}
    See in the Appendix.
\end{proof}
For the sake of completeness, a brief proof of Lemma \eqref{thm_LWE_PRM} is given in the appendix.

We now proceed to describe the proposed multi-recipient encryption scheme based on the pseudorandom map mentioned in Lemma \eqref{thm_LWE_PRM}.
\subsection{Multi-Key Multi-Recipient (MKMR) Encryption Scheme}
\label{sec_enc_scheme}
The  multi-recipient encryption scheme described in this subsection is a symmetric encryption scheme involving multiple messages encrypted using independent secret keys to generate a single ciphertext stream. This stream is then decrypted by the recipients using their respective secret keys to recover the message intended for them. Let the number of recipients be $m$. Corresponding to each recipient there is a unique secret key $s_k \xleftarrow{\scriptscriptstyle{\$}} \mathbb{F}_q^m$ for $k = 1,2,\ldots,m$.  The proposed encryption scheme consists of the following steps:


\begin{itemize}
    \item \textbf{Setup($\bm{\lambda}$):} This algorithm takes the security parameter $\lambda$ as input and returns an  integer $m$, a prime $q$ and a discrete Gaussian distribution $\chi$ over $\mathbb{F}_q$.
    \item \textbf{KeyGen($m,q$):} 
    \label{alg:key_gen}
    The algorithm \textit{KeyGen(.)} takes $m$ and $q$ generated by \textbf{Setup($\bullet$)} and outputs $m$ secret keys $s_1,s_2,\ldots,s_m$ sampled from a uniform distribution over $\mathbb{F}_q^m$.
    
    \item \textbf{Encrypt} (denoted as \textit{\textbf{Enc$\bm{\left(s_1,s_2,\ldots,s_m, \mathcal{M} \right)}$}}):
    Consider a $m$ recipient encryption scheme where each recipient has a message stream of length $l$. The message stream corresponding to the $j^{th}$ recipient is $\begin{bmatrix}
        m_{j1} & m_{j2} & \ldots & m_{jl} \\
        \end{bmatrix} \in \mathbb{F}_q^l$. Now stacking these message streams together gives a message matrix $\mathcal{M} \in \mathbb{F}_q^{m \times l}$ which is represented as 
    \begin{equation}
        \mathcal{M} = \begin{bmatrix}
            m_{11} & m_{12} & \ldots & m_{1l} \\
            m_{21} & m_{22} & \ldots & m_{2l} \\
            \vdots & \vdots & \ddots  & \vdots \\
            m_{m1} & m_{m2} & \ldots & m_{ml} \\
            \end{bmatrix} 
    \end{equation}  
       Let $m_i$ be the $i^{th}$ column of the matrix $\mathcal{M}$. The ciphertext is a sequence of vectors $v_0,v_1,\ldots,v_{l}$ in $\mathbb{F}_q^m$. The vector $v_0$ is the initialization vector randomly sampled from a uniform distribution over $\mathbb{F}_q^m$. The remaining vectors are recursively generated as follows,
        \begin{align}
        \label{mkmr_ciphertext_vector}
            v_{i} &=  m_{i} + f_{\chi}\left(S,v_{i-1}\right)\,\,\left(\text{mod}\,\,q\right) \notag \\
            &= m_i + \begin{bmatrix}
        \left \langle s_1,v_{i-1} \right \rangle + e_1 \\
        \left \langle s_2,v_{i-1} \right \rangle + e_2 \\
        \left \langle s_3,v_{i-1} \right \rangle + e_3 \\
        \vdots  \\
        \left \langle s_m,v_{i-1} \right \rangle + e_m \\
        \end{bmatrix}\,\,\left(\text{mod}\,\,q\right)
        \end{align}
where $E_i = [e_1,e_2,\ldots,e_m]^T \in \mathbb{F}_q^m$ is randomly sampled from the distribution $\chi$. The encryption procedure is described in Algorithm \eqref{alg:ciphertext_algo}.
\begin{algorithm}[H]
\caption{\it Enc($\bullet$)}\label{alg:ciphertext_algo}
\begin{algorithmic}[1]
\STATE \textbf{Input:}\,\,$s_1,s_2,\ldots,s_m, \mathcal{M} = \begin{bmatrix}
        m_{1} & m_{2} & \ldots & m_{l} \\
        \end{bmatrix}$
\STATE \textbf{Output:} $\mathcal{C}$ \\
\STATE $I_v \xleftarrow{\scriptscriptstyle{\$}} \mathbb{F}_q^{m}$
\STATE Set $v_0 = I_v$ \\

    \STATE Set $i = 1$ \\
    \WHILE{$i < l$}
    
        \STATE $E_i\xleftarrow{\scriptscriptstyle{\$}} \chi^m$
        \STATE $v_{i} =  m_i + \begin{bmatrix}
        \left \langle s_1,v_{i-1} \right \rangle + e_1 \\
        \left \langle s_2,v_{i-1} \right \rangle + e_2 \\
        \left \langle s_3,v_{i-1} \right \rangle + e_3 \\
        \vdots  \\
        \left \langle s_m,v_{i-1} \right \rangle + e_m \\
        \end{bmatrix}\,\,\left(\text{mod}\,\,q\right)$
        \STATE Update $i = i+1$
    \ENDWHILE 
\STATE $\mathcal{C} = \begin{bmatrix}
        v_0 & v_1 & \ldots & v_{l} \\
        \end{bmatrix}$
\STATE \textbf{\textit{return $\mathcal{C}$}}
\end{algorithmic}
\end{algorithm}
    \item \textbf{Decrypt} (denoted as \textit{\textbf{Dec$\bm{\left(s_j, \mathcal{C} \right)}$}}): The decryption process is the reverse of encryption. Given the secret key $s_j$, the $j^{th}$ receiver recursively recovers its message stream $[m_{j1},m_{j2},\ldots,m_{jl}]$ as follows:
    \begin{equation}
        m_{ji}  = v_{ji} -\begin{bmatrix}
        \left \langle s_j,v_{i-1} \right \rangle 
        \end{bmatrix}\,\,\left(\text{mod}\,\,q\right), \textrm{ for }1 \leq i < l
    \end{equation}
    where $v_{ji}$ is the $j^{th}$ entry of the vector $v_i$. It can be easily verified that the recovered message is the same as the original message with some added noise. The decryption process for the $j^{th}$ receiver is given in Algorithm \eqref{alg:decrypt_algo}.
\begin{algorithm}[H]
\caption{\it Dec($\bullet$)}\label{alg:decrypt_algo}
\begin{algorithmic}[1]
\STATE \textbf{Input:}\,\,$s_j,\mathcal{C} = \begin{bmatrix}
        v_0 & v_1 & \ldots & v_{l} \\
        \end{bmatrix}$ \\
\STATE \textbf{Output:} $\mathcal{M}_j = \begin{bmatrix}
        m_{j1} & m_{j2} & \ldots & m_{jl} \\
        \end{bmatrix}$ \\

\STATE Set $i=1$
    \WHILE{$i < l$}
        \STATE $m_{ji}  = v_{ji} -\begin{bmatrix}
        \left \langle s_j,v_{i-1} \right \rangle 
        \end{bmatrix}\,\,\left(\text{mod}\,\,q\right)$
        \STATE Update $i=i+1$
    \ENDWHILE
\STATE $\mathcal{M}_j = \begin{bmatrix}
        m_{j1} & m_{j2} & \ldots & m_{jl} \\
        \end{bmatrix}$
\STATE \textbf{\textit{return $\mathcal{M}_j$}}
\end{algorithmic}
\end{algorithm}
\end{itemize}

We now proceed to analyze the security of the proposed scheme.

\subsection{Security}

We start by defining indistinguishability under chosen plaintext attack (IND-CPA) security for a symmetric multirecipient scheme . The IND-CPA game for a symmetric multirecipient scheme  which is played between an adversary and a challenger has the following stages
\begin{enumerate}
    \item {\it Initialize}: The challenger randomly samples a set of $m$-keys from $\mathbb{F}_q^m$. Further, she also randomly samples a bit $b \in \{0,1\}$.
    \item {\it Querying the challenger (encryption oracle)}: The adversary chooses a set of $p$ $m$-tuples of messages $(\mathcal{M}_1,\mathcal{M}_2,\ldots,\mathcal{M}_p)$ and sends to the challenger. The challenger acts as an encryption oracle, encrypts these message tuples and sends it back to the adversary. Let the corresponding ciphertexts be $\mathcal{C}_1,\mathcal{C}_2,\ldots,\mathcal{C}_p$.
    \item{\it Challenge}: The adversary chooses two $m$-tuples of messages $\mathcal{M}_{01}$ and $\mathcal{M}_{11}$ which are not identical to each other and to any of the message tuples previously queried and sends them to the challenger. Based on the value of $b$, the challenger encrypts $\mathcal{M}_{b1}$ and sends it to the adversary.
    \item{\it Guessing $b$}: The adversary attempts to guess the value of $b$ using the information at her disposal. She wins the game on guessing $b$ correctly.
\end{enumerate}

The IND-CPA security of the proposed scheme follows as a direct consequence of the following two results.

\begin{lemma}
\label{indis}
    Consider a set of $p$ $m$-tuples $\left(\mathcal{M}_1,\mathcal{M}_2,\ldots,\mathcal{M}_p\right)$ of message streams of length $l$, where $p$ is a polynomial in $m$. 
    Let  $\left(\mathcal{C}_1,\mathcal{C}_2,\ldots,\mathcal{C}_p \right)$ be the corresponding ciphertexts obtained  by encrypting the message streams using Algorithm \eqref{alg:ciphertext_algo}. 
    Then, under the $\mathcal{LWE}_{m,q,\chi}$assumption,  $\left(\mathcal{C}_1,\mathcal{C}_2,\ldots,\mathcal{C}_p \right)$  is indistinguishable from a set of $p$ elements each independently randomly sampled from a uniform distribution over $\mathbb{F}_q^{m \times (l+1)}$.
\end{lemma}

\begin{proof}
Consider an algorithm $\mathcal{A}$ that aims to distinguish between $\left(\mathcal{C}_1,\mathcal{C}_2,\ldots,\mathcal{C}_p \right)$ and from a set of $p$ elements each independently randomly sampled from a uniform distribution over $\mathbb{F}_q^{m \times (l+1)}$. Let  $H_0,H_1,\ldots, H_p$ be a set of $p+1$ hybrid experiments that returns $p$ elements of $\mathbb{F}_q^{m \times (l+1)}$. In particular, for $0 \leq i \leq p$, the output of $H_i$ is a $p$-tuple, the first $i$ elements of which are independently randomly sampled from $U(\mathbb{F}_q^{m \times (l+1)})$ and the remaining elements are $\left(\mathcal{C}_{i+1},\mathcal{C}_{i+2},\ldots,\mathcal{C}_p \right)$. Note that $H_0$ returns the ciphertext $\left(\mathcal{C}_1,\mathcal{C}_2,\ldots,\mathcal{C}_p \right)$ and $H_p$ returns $p$ elements each independently randomly sampled from a uniform distribution over $\mathbb{F}_q^{m \times (l+1)}$.

Let $\mathcal{P}_i$ be the probability that $\mathcal{A}$ returns 1 when its input is the output of $H_i$.  We claim that, for any $1 \leq k \leq p$, $|\mathcal{P}_p - \mathcal{P}_{p-i}| \leq il\epsilon(m)$. We prove this using induction. \\
{\em Base case: } Observe that, $\left|\mathcal{P}_p - \mathcal{P}_0 \right|$ is the advantage of the Algorithm $\mathcal{A}$. Now, as a consequence of Theorem \eqref{thm_PRM} and Lemma \eqref{thm_LWE_PRM}, $\left|\mathcal{P}_p - \mathcal{P}_{p-1} \right|$ is less than $ l\epsilon(m)$, where $\epsilon(m)$ is the maximum advantage of an adversary against the pseudorandom map $f_\chi(S,\bullet)$. \\
{\em Induction step: } 
Assume that the claim is true for $i = k$. We have to prove that $|\mathcal{P}_p - \mathcal{P}_{p-(k+1)}| \leq (k+1)l\epsilon(m)$. \\
Observe that, distinguishing between the outputs $H_{p-k}$ and $H_{p-(k+1)}$ is equivalent to distinguishing between $C_{p-(k+1)}$ and a randomly sampled element of $\mathbb{F}_q^{m \times (l+1)}$. Therefore, $|\mathcal{P}_{p-(k+1)} - \mathcal{P}_{p-k}| \leq  l\epsilon(m)$. Now,
\begin{eqnarray*}
    |\mathcal{P}_p - \mathcal{P}_{p-(k+1)}| &\leq& |\mathcal{P}_p - \mathcal{P}_{p-k} +\mathcal{P}_{p-k} -\mathcal{P}_{p-(k+1)}| \\
    &\leq& |\mathcal{P}_p - \mathcal{P}_{p-k}| +|\mathcal{P}_{p-k} -\mathcal{P}_{p-(k+1)}|\\
    &\leq&  kl\epsilon(m) + l\epsilon(m) = (k+1)l\epsilon(m)
\end{eqnarray*}
Hence, the claim is proved. \\
Therefore $|\mathcal{P}_p - \mathcal{P}_{0}| \leq pl\epsilon(m)$. As $p$
is polynomial in $m$ and $\epsilon$ is a negligible function of $m$, $(\mathcal{C}_1,\mathcal{C}_2,\ldots,\mathcal{C}_p)$ is  indistinguishable from a set of $p$ elements each independently randomly sampled from a uniform distribution over $\mathbb{F}_q^{m \times (l+1)}$.\end{proof}
We now prove that the encryptions of two distinct $m$-tuples of message streams are indistinguishable.

\begin{lemma}
\label{indis1}
Consider two $m$-tuples of messages $\mathcal{M}_{01}$ and $\mathcal{M}_{02}$ wherein each message is of length $l$ where $l$ is bounded by a polynomial in $m$.    
Let $\mathcal{C}_{01} = \left[v_0, v_{1},v_{2},\ldots,v_{l} \right] \in \mathbb{F}_q^{m \times (l+1)}$ and $\mathcal{C}_{02} = \left[v_o', v_{1}',v_{2}',\ldots,v_{l}' \right] \in \mathbb{F}_q^{m \times (l+1)}$ be the ciphertexts obtained by respectively encrypting  $\mathcal{M}_{01}$ and $\mathcal{M}_{02}$, using Algorithm \ref{alg:ciphertext_algo}. 
 Then, under the $\mathcal{LWE}_{m,q,\chi}$assumption,   $\mathcal{C}_{01}$ and $\mathcal{C}_{02}$ are indistinguishable from each other and a randomly sampled element of $\mathbb{F}_q^{m \times (l+1)}$.
\end{lemma}
\begin{proof}
Consider an algorithm $\mathcal{A}$ with a binary output that takes inputs from $\mathbb{F}_q^{m \times (l+1)}$.

Let $\mathcal{P}_1$ and $\mathcal{P}_2$ be the probabilities of $\mathcal{A}$ returning $1$ for inputs $ \mathcal{C}_{01}$ and $\mathcal{C}_{02}$ respectively.
\begin{eqnarray*}
    \mathcal{P}_1 &=& Pr\left[\mathcal{A}\left( \mathcal{C}_{01} \right) = 1  \right]\\
    \mathcal{P}_2 &=& Pr\left[\mathcal{A}\left( \mathcal{C}_{02} \right) = 1  \right]
\end{eqnarray*}
Let $\mathcal{P}_0$ be the probability of $\mathcal{A}$ returning $1$ when the input is randomly sampled from the distribution $U(\mathbb{F}_q^{m \times l})$.
\begin{equation}
    \mathcal{P}_0 = Pr\left[\mathcal{A}\left( \mathcal{C}_0 \right) = 1  \right]
\end{equation}
From Theorem \ref{thm_PRM} and Lemma \ref{thm_LWE_PRM}, $\left|\mathcal{P}_{1} - \mathcal{P}_{0} \right|$ and $\left|\mathcal{P}_{2} - \mathcal{P}_{0} \right|$ are less than $l\epsilon(m)$ where $\epsilon(m)$ is a negligible function of the security parameter. Let their values be $\epsilon_1$ and $\epsilon_2$, respectively.

Now, the advantage of $\mathcal{A}$ in distinguishing between $\mathcal{C}_{01}$ and $\mathcal{C}_{02}$ is given by
\begin{align}
        Adv(\mathcal{A}) &= \left|\mathcal{P}_{1} - \mathcal{P}_{2} \right| \notag \\
        &= \left|\mathcal{P}_{1} - \mathcal{P}_{0} + \mathcal{P}_{0} - \mathcal{P}_{2} \right| \notag \\
        &\leq \left|\mathcal{P}_{1} - \mathcal{P}_{0} \right| + \left|\mathcal{P}_{0} - \mathcal{P}_{2} \right| \notag \\
        &\leq \left|\mathcal{P}_{1} - \mathcal{P}_{0} \right| + \left|\mathcal{P}_{2} - \mathcal{P}_{0} \right| \notag \\
        &\leq l\epsilon(m) + l\epsilon(m) = 2l\epsilon(m)
\end{align}
Hence, $Adv(\mathcal{A})$  is negligible. Hence, ciphertext streams $\mathcal{C}_{01}$ and $\mathcal{C}_{02}$ are indistinguishable from each other.

\end{proof}

Using the preceding couple of lemmas, the following theorem establishes the IND-CPA security of the proposed scheme.
\begin{thm}\label{SecProof1}
    Consider a set of $p$ $m$-tuples $\left(\mathcal{M}_1,\mathcal{M}_2,\ldots,\mathcal{M}_p\right)$ of message streams of length $l$, where $p$ is a polynomial in $m$. 
    Let  $\left(\mathcal{C}_1,\mathcal{C}_2,\ldots,\mathcal{C}_p \right)$ be the corresponding ciphertexts obtained  by encrypting the message streams with an  $m$-tuple secret key $S$ using Algorithm \eqref{alg:ciphertext_algo}.     
    Assume two $m$-tuples of messages $\mathcal{M}_{01}$ and $\mathcal{M}_{02}$ that are distinct from each other and each of the message streams $\left(\mathcal{M}_1,\mathcal{M}_2,\ldots,\mathcal{M}_p\right)$. 
    Further, let  $\mathcal{C}_{01}$ and $\mathcal{C}_{02}$ be the ciphertexts of $\mathcal{M}_{01}$ and $\mathcal{M}_{02}$, respectively, obtained using the secret key $S$ in Algorithm \eqref{alg:ciphertext_algo}.     
    Then, under the $\mathcal{LWE}_{m,q,\chi}$ assumption, no algorithm with access to the pairs $(\mathcal{M}_i,\mathcal{C}_i)$ for $1\leq i\leq p$ can distinguish between $\mathcal{C}_{01}$ and $\mathcal{C}_{02}$ 
    with any significant advantage i.e., the encryption scheme is IND-CPA secure.
\end{thm}
\begin{proof}
    Let $\Delta$ be the advantage of an algorithm $\mathcal{A}$, with access to the pairs $(\mathcal{M}_i,\mathcal{C}_i)$ in  distinguishing between $\mathcal{C}_{01}$ and $\mathcal{C}_{02}$. 
    To the contrary, assume that $\Delta$ is not negligible in $m$. Now, if the $\mathcal{C}_i$s that are fed to the algorithm $\mathcal{A}$ are replaced by randomly sampled elements of $\mathbb{F}_q^{m \times (l+1)}$, then, by Lemma \eqref{indis1}, the advantage of $\mathcal{A}$ is less than $2l\epsilon(m)$ where $\epsilon$ is a negligible function of $m$.

    Now, the algorithm $\mathcal{A}$ can be used to create an algorithm $\mathcal{A}'$ that can distinguish between the $p$-tuple $\mathcal{C}_1,\mathcal{C}_2,\ldots,\mathcal{C}_p$ from a set of $p$ randomly sampled elements of $\mathbb{F}_q^{m \times (l+1)}$ with significant advantage. Given access to both the challenger and the adversary algorithms in the IND-CPA game and a set of $p$ message-ciphertext pairs $(\mathcal{M}_i,\mathcal{C}_i)$ $(1\leq i\leq p)$, $\mathcal{A}'$ can instruct the adversary $\mathcal{A}$ to generate a pair of messages $\mathcal{M}_{01},\mathcal{M}_{02}$ and the challenger to generate a challenge for the adversary as in the IND-CPA game. The algorithm $\mathcal{A}'$ returns 1 if $\mathcal{A}$ guesses the challenge message correctly. Clearly, the advantage of $\mathcal{A}'$ is $\Delta - 2l\epsilon(m)$ which is not negligible. This contradicts Lemma \eqref{indis}. Hence, $\Delta$ is a negligible function of $m$ and the proposed scheme is IND-CPA secure.
\end{proof}

The above notion of IND-CPA security ensures the security of the encrypted data against an adversary who does not have access to the secret keys. However, in a multi-recipient scheme, it is important that any receiver or set of receivers should not be able to extract information about message streams that are not intended for them. We now extend the idea to IND-CPA security to the case when the adversary knows some of the keys. Assume that the adversary knows $k$ of the keys. Without loss of generality, we can assume these to be the last $k$ keys. We now proceed to prove that if, in addition to the $k$ keys, the adversary also knows the last $k$ entries of the other keys, she can still not distinguish between the first $m-k$  rows of encryptions of two distinct messages.

\begin{thm}\label{SecProof2}
	Consider a set of $p$ $m$-tuples $\left(\mathcal{M}_1,\mathcal{M}_2,\ldots,\mathcal{M}_p\right)$ of message streams of length $l$, where $p$ is a polynomial in $m$. 
	Let $\left(\mathcal{C}_1,\mathcal{C}_2,\ldots,\mathcal{C}_p \right)$ be the corresponding ciphertexts obtained by encrypting the message streams with a secret key matrix $S=[s_1,s_2,\ldots,s_m]^T$ using Algorithm \eqref{alg:ciphertext_algo}. (The rows of $S$ constitute the $m$ secret keys).
     Let $\mathcal{M} = (m_1,m_2,\ldots,m_l)$ be an arbitrary message that is distinct from $\left(\mathcal{M}_1,\mathcal{M}_2,\ldots,\mathcal{M}_p\right)$ and let $\mathcal{C}$ be the corresponding ciphertext obtained using $S$ in Algorithm \eqref{alg:ciphertext_algo}. 
    Consider an adversary $\mathcal{A}$ who has access
    to the following: 
    \begin{itemize}
    \item[(i)] pairs $(\mathcal{M}_i,\mathcal{C}_i)$ for $1\leq i \leq p$. 
    \item[(ii)] The last $k$ secret keys from $S$, and 
    \item[(iii)] The last $k$ entries of the remaining ${m-k}$ secret keys in $S$. 
    \item[(iv)] An oracle that evaluates the following map $f$ from $\mathbb{F}_q^{m-k}$ to $\mathbb{F}_q^{m-k}$.
    \begin{equation*}
        f(v) = S_1v + \epsilon,
    \end{equation*}
    where $\epsilon$ is randomly sampled from $\chi^{m-k}$ and $S_1$ is the matrix consisting of the first $m-k$ entries of the first $m-k$ rows of $S$ i.e. the top left $(m-k) \times (m-k)$ submatrix of $S$.
   \end{itemize}
	Then, under the $\mathcal{LWE}_{(m-k),q,\chi}$ assumption,  the adversary cannot efficiently distinguish the first $m-k$ rows of $\mathcal{C}$ from a matrix randomly sampled from a uniform distribution over $\mathbb{F}_q^{(m-k) \times (l+1)}$.
\end{thm}

\begin{proof}
	  We partition $S \in \mathbb{F}_q^{m\times m}$ as $S = \left[\begin{array}{c|c} S_1 & S_2\\ \hline S_3 & S_4\end{array}\right]$, where $S_1 \in \mathbb{F}_q^{(m-k)\times(m-k)}$. Therefore,  $\mathcal{A}$  has access to $S_2, S_3$, and $S_4$. 
 	We now prove that even with access to $S_2, S_3$ and $S_4$, the first $m-k$ rows of $\mathcal{C}$ remain indistinguishable from  a matrix randomly sampled from a uniform distribution over $\mathbb{F}_q^{(m-k) \times (l+1)}$. 

	Consider $l+1$ hybrid experiments $H_i$ for $1 \leq i \leq l+1$. We define the output of the experiments next. 
	Corresponding to an experiment $H_i$, partition $\mathcal{C}$ as follows:
	\begin{align}\label{eqn:cipher}
		\mathcal{C} = 
		\left[\begin{array}{c|c} 
			\mathcal{C}_{1i} & \mathcal{C}_{2}\\ \hline \mathcal{C}_{3i} & \mathcal{C}_{4i}
		\end{array} \right] \in \mathbb{F}_q^{m \times (l+1)}, \mbox{ where } \mathcal{C}_{1i} \in \mathbb{F}_q^{(m-k) \times i}
	\end{align}	
	We define a new matrix 	$\widehat{\mathcal{C}}_i$ using the matrix in eq. \eqref{eqn:cipher} as follows
	\begin{align*}
		\widehat{\mathcal{C}}_i := 
		\left[\begin{array}{c|c} 
			\mathcal{C}_{1i} & \mathcal{U}_{i}\\ \hline \mathcal{C}_{3i} & \mathcal{C}_{4i}
		\end{array} \right] \in \mathbb{F}_q^{m \times (l+1)}, \mbox{ where } \mathcal{C}_{1i} \in \mathbb{F}_q^{(m-k) \times i}
	\end{align*}
	$\mathcal{U}_{i}$ are randomly sampled elements from $\mathbb{F}_q^{(m-k) \times (l+1-i)}$. The output of $H_i$ is the first $m-k$ rows of $\widehat{\mathcal{C}}_i$, i.e., $\begin{bmatrix}\mathcal{C}_{1i} & \mathcal{U}_i\end{bmatrix}$. 
    Hence, we need to prove that the output of the experiments $ H_1, \cdots, H_{l+1}$ are indistingushable from a matrix randomly sampled from a uniform distribution over $\mathbb{F}_q^{(m-k) \times (l+1)}$. We prove this using induction. 
    
    {\em Base case: } Note that the output of $H_1$ experiment are 
    the first $m-k$ rows of $\widehat{\mathcal{C}}_1=\begin{bmatrix}v_0 & \mathcal{U}_1\end{bmatrix}$. Since  $v_0$ is randomly sampled in the encryption algorithm, the output of $H_1$ is identical to a randomly sampled element from $U(\mathbb{F}_q^{(m-k)\times(l+1)})$. 
	
%
%
%
%
%
%
%
%
    {\em Induction step: } Assume that the outputs of the experiments $H_1,H_2,\ldots,H_t$ are indistinguishable from each other and from randomly sampled elements of $\mathbb{F}_q^{(m-k) \times (l+1)}$. We now proceed to prove that this also holds true for the output of $H_{t+1}$. 

	 Let $P_i$ be the probability that $\mathcal{A}$ outputs 1 when its input is the output of $H_i$, for $1 \leq i \leq l$. Let 
	\begin{align*}
		\widehat{\mathcal{C}}_{t-1}\!\!=\!\! 
		\left[\begin{array}{c|c}
			\mathcal{C}_{1(t-1)} & \mathcal{U}_{t-1}\\ \hline 
			\mathcal{C}_{3(t-1)} & \mathcal{C}_{4(t-1)}
		\end{array}\right]  \!\!= \!\!
		\left[\begin{array}{ccc|cccc}
			\bar{v}_0 \!\!&\!\!\ \!\!\cdots \!\!\!\!&\!\! \bar{v}_{t-2} & \bar{w}_{t-1} & \bar{w}_{t} \!\!&\!\!\!\! \cdots \!\!&\!\! \bar{w}_{l}\\ \hline 
			\bar{\bar{v}}_0 \!\!&\!\!\!\!\cdots\!\! \!\!&\!\!\bar{\bar{v}}_{t-2} & \bar{\bar{v}}_{t-1} \!\!&\!\!\!\! \bar{\bar{v}}_{t} \!\!&\!\! \cdots & \bar{\bar{v}}_{l}
		\end{array}\right]
	\end{align*}
	Replace $\bar{w}_{t}$ with $\tilde{w}_{t} = \bar{m}_t + f(\bar{w}_{t-1}) + S_2\bar{\bar{v}}_{t-1} = \bar{m}_t+S_1\bar{w}_{t-1} + S_2\bar{\bar{v}}_{t-1} + e$ where $e$ is randomly sampled from $\chi^{m-k}$ and $\bar{m}_t$ consists of the first $m-k$ entries of $m_t$. Define 
	\begin{align*}
		 \mathfrak{K}_{t-1} := 	\left[\begin{array}{ccc|cccc}
			\bar{v}_0 \!\!&\!\!\ \!\!\cdots \!\!\!\!&\!\! \bar{v}_{t-2} & \bar{w}_{t-1} & \tilde{w}_{t} \!\!&\!\!\!\! \cdots \!\!&\!\! \bar{w}_{l}\\ \hline 
			\bar{\bar{v}}_0 \!\!&\!\!\!\!\cdots\!\! \!\!&\!\!\bar{\bar{v}}_{t-2} & \bar{\bar{v}}_{t-1} \!\!&\!\!\!\! \bar{\bar{v}}_{t} \!\!&\!\! \cdots & \bar{\bar{v}}_{l}
		\end{array}\right]
	\end{align*}
	Let $\tilde{\mathcal{P}}_{t-1} $ be the probability that $\mathcal{A}$ returns 1 when the input is the first $m-k$ rows of $\mathfrak{K}_{t-1}$.  Now, as $S_2$ is known to the adversary, distinguishing between $\mathfrak{K}_{t-1}$ and the output of $H_{t-1}$ is equivalent to solving an instance of the $\mathcal{LWE}_{(m-k),q,\chi}$ problem. Therefore $\epsilon_1 = |\tilde{\mathcal{P}}_{t-1} - \mathcal{P}_{t-1}|$ is a negligible function of $m-k$. 

    Similarly, let
\begin{align*}
		\widehat{\mathcal{C}}_{t}\!\!=\!\! 
		\left[\begin{array}{c|c}
			\mathcal{C}_{1t} & \mathcal{U}_{t}\\ \hline 
			\mathcal{C}_{3t} & \mathcal{C}_{4t}
		\end{array}\right]  \!\!= \!\!
		\left[\begin{array}{ccc|cccc}
			\bar{v}_0 \!\!&\!\!\ \!\!\cdots \!\!\!\!&\!\! \bar{v}_{t-1} & \bar{w}_{t}' & \bar{w}_{t+1}' \!\!&\!\!\!\! \cdots \!\!&\!\! \bar{w}_{l}'\\ \hline 
			\bar{\bar{v}}_0 \!\!&\!\!\!\!\cdots\!\! \!\!&\!\!\bar{\bar{v}}_{t-1} & \bar{\bar{v}}_{t} \!\!&\!\!\!\! \bar{\bar{v}}_{t+1} \!\!&\!\! \cdots & \bar{\bar{v}}_{l}
		\end{array}\right]
	\end{align*}

    Replace $\bar{w}_{t}'$ by $\tilde{w}_t' =  \bar{m}_t + f(\bar{v}_{t-1}) + S_2\bar{\bar{v}}_{t-1} =\bar{m}_t+ S_1\bar{v}_{t-1} + S_2\bar{\bar{v}}_{t-1} + e'$ where $e'$ is randomly sampled from $\chi^{m-k}$. Define 
	\begin{align*}
		 \mathfrak{K}_{t} := 	\left[\begin{array}{ccc|cccc}
			\bar{v}_0 \!\!&\!\!\ \!\!\cdots \!\!\!\!&\!\! \bar{v}_{t-1} & \tilde{w}_{t}' & \bar{w}_{t+1}' \!\!&\!\!\!\! \cdots \!\!&\!\! \bar{w}_{l}'\\ \hline 
			\bar{\bar{v}}_0 \!\!&\!\!\!\!\cdots\!\! \!\!&\!\!\bar{\bar{v}}_{t-1} & \bar{\bar{v}}_{t} \!\!&\!\!\!\! \bar{\bar{v}}_{t+1} \!\!&\!\! \cdots & \bar{\bar{v}}_{l}
		\end{array}\right]
	\end{align*}

    Let $\tilde{\mathcal{P}}_{t} $ be the probability that $\mathcal{A}$ returns 1 when the input is the first $m-k$ rows of $\mathfrak{K}_{t}$.
    Observe that the adversary can access the oracle to generate $\mathfrak{K}_{t-1}$ and $\mathfrak{K}_{t}$ from $\widehat{\mathcal{C}}_{t-1}$ and $\widehat{\mathcal{C}}_{t}$ respectively. Therefore, if the adversary can distinguish between $\mathfrak{K}_{t-1}$ and $\mathfrak{K}_{t}$, then he can effectively distinguish between $\widehat{\mathcal{C}}_{t-1}$ and $\widehat{\mathcal{C}}_{t}$. This contradicts the induction assumption. Therefore $|\tilde{\mathcal{P}}_{t-1} - \tilde{\mathcal{P}}_{t}| = \epsilon_2$ is negligible. Now,
    \begin{eqnarray*}
        |\tilde{\mathcal{P}}_{t}-\mathcal{P}_{t-1}| &=& |\tilde{\mathcal{P}}_{t}-\tilde{\mathcal{P}}_{t-1}+\tilde{\mathcal{P}}_{t-1}-\mathcal{P}_{t-1}|\\
        &\leq& |\tilde{\mathcal{P}}_{t}-\tilde{\mathcal{P}}_{t-1}|+|\tilde{\mathcal{P}}_{t-1}-\mathcal{P}_{t-1}|\\
        &=& \epsilon_1 +\epsilon_2.
    \end{eqnarray*}
Hence, $|\tilde{\mathcal{P}}_{t}-\mathcal{P}_{t-1}|$ is a negligible function of $m-k$. 

Observe that the distribution of $\mathfrak{K}_t$ is identical to that of $\widehat{\mathcal{C}}_{t+1}$. Therefore, $\tilde{\mathcal{P}}_{t} = \mathcal{P}_{t+1} $. Hence, $|{\mathcal{P}}_{t+1}-\mathcal{P}_{t-1}|$ is a negligible function of $m-k$. Hence the output of $H_{t+1}$ is indistinguishable from that of $H_{t-1}$ and therefore from the outputs of $H_1, H_2,\ldots, H_t$.

Hence, by induction, the outputs of $H_{l+1}$ and $H_1$ are indistinguishable from each other and a randomly sampled element of $\mathbb{F}_q^{(m-k) \times (l+1)}$.

\end{proof}
The following subsection discusses the choice of parameters based on the desired level of security.

\subsection{Choice of Parameters}
The work of Biasioli et al. \cite{biasioli2024tool}, provides a framework for estimating the desired security parameter $\lambda$ while choosing dimension $m$ in the LWE problem. For different choices of $\lambda$, different values of $m$ are chosen. From \cite[Section 6]{biasioli2024tool}, typically for a $128$-bit security level, i.e. $\lambda = 128$, the LWE dimension is chosen to be $m(\lambda) = 1024$, and the value of parameter $q$ is chosen such that $log\,q = 20-41$.

\subsection{Example}
\label{example}
We now demonstrate how the proposed algorithm can be used to encrypt images. The length of the secret vectors is 1024. The value of $q$ is chosen as $2^{31}-1$. As this number is a Mersenne prime, the $mod\,\,q$ equivalent of any integer can be found very efficiently.

 As the proposed algorithm operates over the field $\mathbb{F}_q$, the plain-text images are first converted to a string of elements in $\mathbb{F}_q$ as follows.

\subsubsection*{Generating an {$\mathbb{F}_q$} string }
Consider a gray-scale image of size $({r} \times {c})$ (having ${r}$ rows of ${c}$ pixels each) where each pixel is encoded as an $8$ bit word. To encrypt the image, we consider a window of  $t = \frac{\left \lfloor log_2q \right \rfloor}{8}$ consecutive pixels on the same row. Each window contains less than $log_2q$ bits and the corresponding integer value can be considered as an element of $\mathbb{F}_q$. The window traverses each row of the image from left to right, shifting one pixel at a time, and then proceeds to the next row. For each row, the first window position contains the first $t$ pixels of that row (starting from the left). As the window approaches the right end of a row, it wraps around the first few pixels from the left. For example, the last window position in any row contains the rightmost pixel of that row and the first $t-1$ pixels of that row. Thus, each row has $y$ window positions and each pixel is contained in $t$ of them. The total number of window positions is the same as the number of pixels in an image i.e., $({rc})$. 

\begin{figure}[ht!]
\centering
\includegraphics[width=0.5\textwidth]{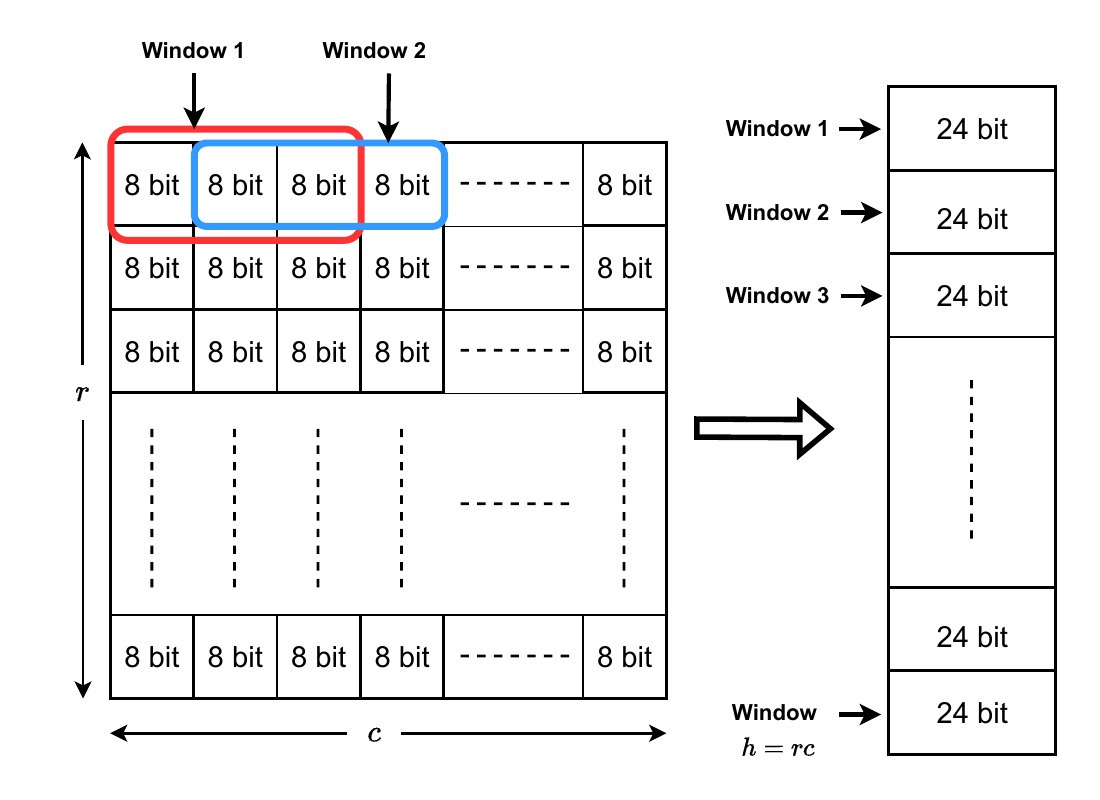}
\caption{A gray-scale image of size $({r} \times c)$ (on the left) transformed to a message stream of length $h$ (on the right)}
\label{image_pixel_example}
\end{figure}

 \subsubsection*{An MKMR Example}
To implement the MKMR scheme on images with the chosen parameters, we need a set of 1024 images. As it is impossible to demonstrate the results of such encryption in this article, we have considered a set of 4 gray-scale images shown in Fig. \eqref{fig_first_case}. These images are converted to 4 streams over $\mathbb{F}_q$ as described above. The other messages are replaced by random strings. The time taken for the encryption and decryption processes is therefore equal to the time needed for encrypting and decrypting 1024 images. Each image is consists of $512 \times 512$ pixels. Thus the total size of data encrypted in the process is approximately 260 MB. The corresponding secret vectors are randomly sampled from a uniform distribution over $\mathbb{F}_q^m$. The window size in this case is taken as 3. Thus each window has 24 bits and is stored as a 32-bit signed integer. During decryption, we first recover the sequence of windows. Each pixel can be recovered from 3 windows. We determine the value of the pixel by considering the window where it occurs in the second position. This is because the central pixel in any window is less effected by noise. The plain-text images, encrypted images and decrypted images are shown in Fig.  \eqref{fig_first_case}, \eqref{fig_second_case} and \eqref{fig_third_case} respectively. The total time taken for the encryption process (not considering the time taken to convert the images into message streams) is 4 seconds on a workstation with an INTEL i5, 2.6 GHz processor, 16 GB RAM, and an Windows-11 operating system. 
\begin{figure*}[ht!]
\centering
\includegraphics[width=\textwidth]{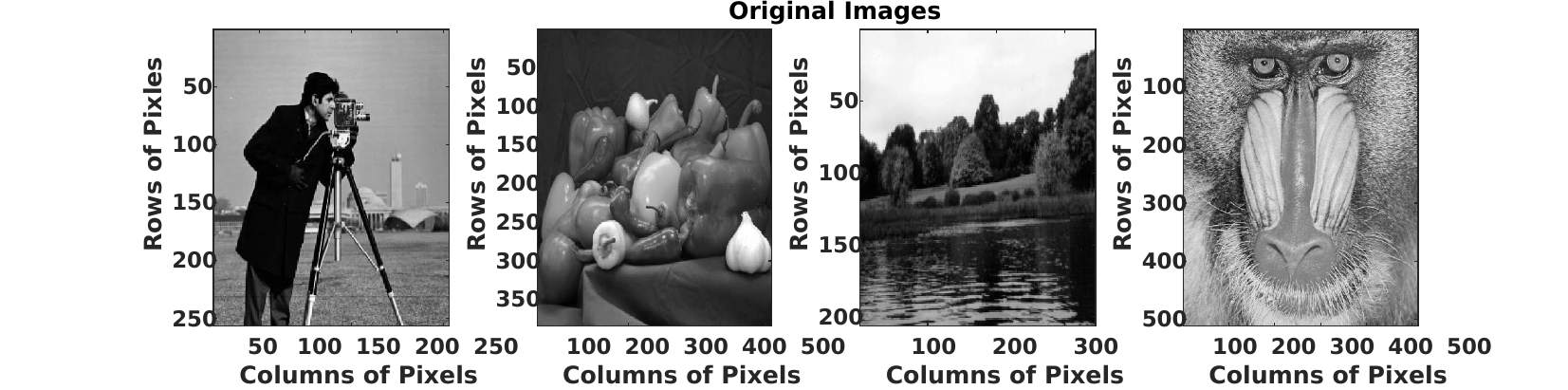}
\caption{Simulation plot of Original Images}
\label{fig_first_case}
\end{figure*}
\begin{figure*}[ht!]
\centering
\includegraphics[width=\textwidth]{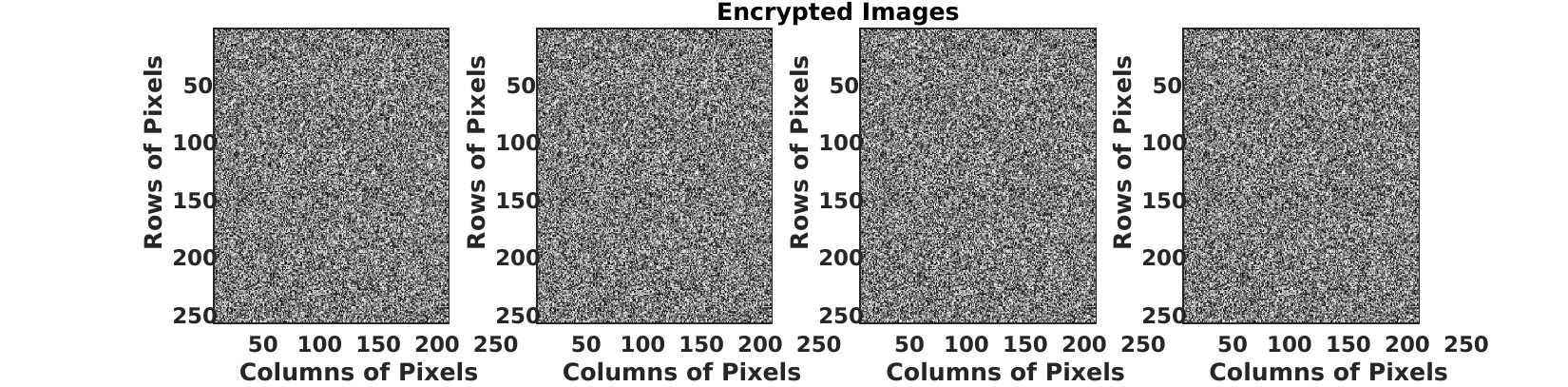}
\caption{Simulation plot of Encrypted Images}
\label{fig_second_case}
\end{figure*}
\begin{figure*}[ht!]
\centering
\includegraphics[width=\textwidth]{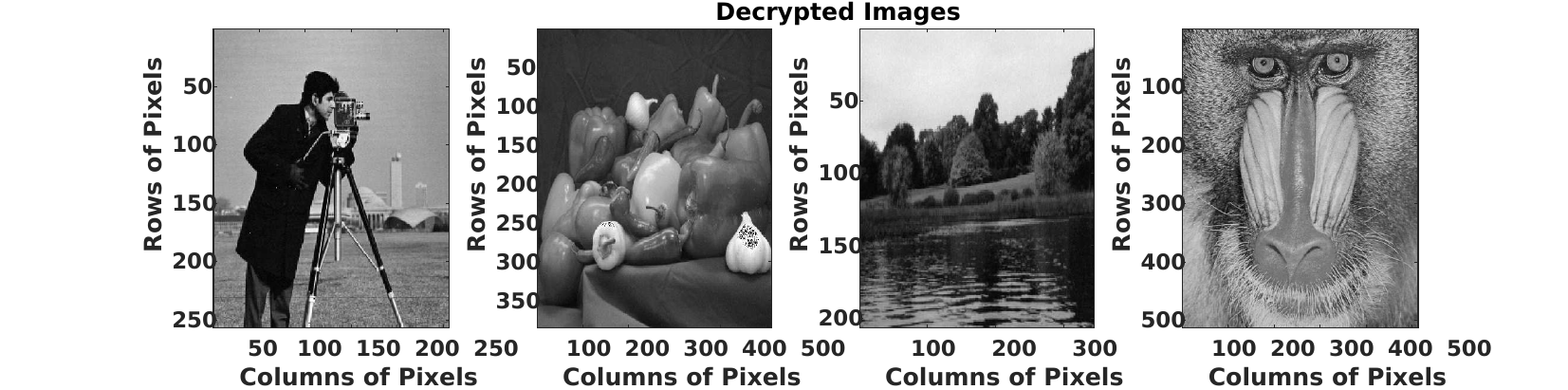}
\caption{Simulation plot of Decrypted Images}
\label{fig_third_case}
\end{figure*}

\section{Conclusion}
This paper presents a scheme for lattice-based multi-recipient symmetric key cryptography.  The proposed scheme considers the case when a single sender communicates with multiple receivers through a single ciphertext. The proposed method aims to achieve high speed encryption while allowing for some addition of noise in the recovered message. The encryption and decryption speeds of these implementations are promising and it will be interesting to explore hardware implementations for higher speeds and efficiency.

{\appendix[Proof of Lemma \ref{thm_LWE_PRM}]
\begin{proof}
Construct a series of hybrid experiments $H_0,H_1,\ldots,H_m$ where in the experiment $H_k$, the elements $(v_1,v_2,\ldots,v_k)$ are obtained using equation \eqref{f_S_prm} where the error terms $e_1,e_2,\ldots,e_k$ are randomly sampled from uniform distribution over $\mathbb{F}_q$ and $v_i$ for $i = k+1,k+2,\ldots,m$ are calculated using equation \eqref{f_S_prm} where the error terms $e_{k+1},e_{k+2},\ldots,e_m$ are sampled from distribution $\chi$.

    The output of $H_k$ is $(v_1,v_2,\ldots,v_m) \in \mathbb{F}_q^{m}$. Observe that the experiment $H_0$ corresponds to the output $(v_1,v_2,\ldots,v_m)$ where the error terms are sampled from distribution $\chi$ and $H_m$ corresponds to the output $(v_1,v_2,\ldots,v_m)$ where the error terms are randomly sampled from uniform distribution.\par
    Consider a binary algorithm $\mathcal{A}$ that accepts elements from $\mathbb{F}_q^{m}$ and aims to distinguish between experiments $H_0$ and $H_m$. Let 
    \begin{equation}
        \mathcal{P}_k = Pr\left[\mathcal{A}(H_k) = 1 \right]
    \end{equation}
    be the probability that $\mathcal{A}$ returns $1$ when it takes elements generated from the experiment $H_k$ as the input. If for any $k$, algorithm $\mathcal{A}$ can distinguish between samples from $H_{k-1}$ and $H_k$, then $\mathcal{A}$ can be used to construct an adversary $\mathcal{A}'$, with oracle access to $f_{\chi}(s,\bullet)$ that can distinguish outputs of $f_{\chi}(s,\bullet)$ from that of a truly random map. But under the $\mathcal{LWE}_{m,q,\chi}$ assumption, the advantage of $\mathcal{A}'$ i.e. $Adv(\mathcal{A}')$ is negligible function of $m$. So, $Adv(\mathcal{A})$ is also negligible function of $m$. Hence, the map $f_{\chi}$ is indistinguishable from a truly random map in $\mathbb{F}_q^m$. Therefore, $f_{\chi}$ is a pseudorandom map.

\end{proof}


 
%

\bibliographystyle{IEEEtran}
\bibliography{bib_file.bib}












\newpage

 




\vfill

\end{document}